\documentclass[
twocolumn,
showpacs,preprintnumbers,amsmath,amssymb,aps,pra,
superscriptaddress,
nofootinbib
]{revtex4-1}

\usepackage{graphicx}
\usepackage{bm}
\usepackage{amsfonts}
\usepackage{amscd}
\usepackage{amsmath}    
\usepackage{enumerate}
\usepackage{subfigure}

\usepackage{amssymb}
\usepackage{amsthm}

\usepackage[breaklinks]{hyperref}

\newtheorem{theorem}{Theorem}

\newtheorem{definition}{Definition}

\newtheorem{remark}{Remark}

\newcommand{\bra}[1]{\mbox{$\left\langle #1 \right|$}}
\newcommand{\ket}[1]{\mbox{$\left| #1 \right\rangle$}}

\begin{document}

\title{Quantum key distribution with delayed privacy amplification 
and its application to security proof of a two-way deterministic protocol}

\author{Chi-Hang Fred Fung}
\email{chffung@hku.hk}
\affiliation{Department of Physics and Center of Computational and Theoretical
 Physics, University of Hong Kong, Pokfulam Road, Hong Kong}

\author{Xiongfeng Ma}
\email{xfma@iqc.ca}
\affiliation{%
Center for Quantum Information and Quantum Control,\\
Department of Physics and Department of Electrical \& Computer Engineering,\\
University of Toronto, Toronto,  Ontario, Canada}
\author{H.~F. Chau}
\email{hfchau@hku.hk}
\affiliation{Department of Physics and Center of Computational and Theoretical
 Physics, University of Hong Kong, Pokfulam Road, Hong Kong}
\author{Qing-yu Cai}
\email{qycai@wipm.ac.cn}
\affiliation{State Key Laboratory of
Magnetics Resonances and Atomic and Molecular Physics, Wuhan
Institute of Physics and Mathematics, Chinese Academy of Sciences,
Wuhan 430071, People's Republic of China}

\date{\today}

\begin{abstract}
Privacy amplification (PA) is an essential post-processing step in quantum key distribution (QKD) for removing any information an eavesdropper may have on the final secret key.
In this paper, we consider delaying PA of the final key after its use in one-time pad encryption and prove its security.
We prove that the security and the key generation rate are not affected by delaying PA.
Delaying PA has two applications: 
it serves as a tool for significantly simplifying the security proof of QKD with a two-way quantum channel, and also it is useful in QKD networks with trusted relays.
To illustrate the power of the delayed PA idea,
we use it to prove the security of a qubit-based two-way deterministic QKD protocol 
which uses four states and four encoding operations.
\end{abstract}


\maketitle

\section{Introduction}

Quantum key distribution (QKD)~\cite{Bennett1984,Ekert1991} allows two parties, Alice and Bob, to share a secret key by exchanging quantum particles.
The final secret key is secure against any eavesdropper, called Eve, with unlimited  computational power.
Initial security proofs of QKD mostly focus on infinite key size and perfect equipment~\cite{Mayers2001,Biham2000,Inamori2005,Lo1999,Shor2000,Kraus2005,Renner2005,Renner2005c}.
More recent security proofs take into consideration device imperfection~\cite{Gottesman2004,Koashi2009comple,Tsurumaru2008,Beaudry2008,Fung2009_effmismatch,Lydersen2010}; while
the effect of finite key size is explicitly considered in Refs.~\cite{PhysRevA.76.012329,Hayashi:Finite:06,Scarani_finite_08,Scarani:Finite:2008b,Cai_finite_2008,Fung2010_Finite,Ma2011_Finite}.

QKD protocols usually involve two post-processing steps after the quantum state transmission step: error correction (EC)~\cite{Shannon1948} to make sure Alice's key is the same as Bob's and privacy amplification (PA)~\cite{Bennett:1988:PAP:45474.45477,Bennett1995} to ensure Eve does not have any non-trivial information on the final secret key.
The final secret key generated can then be used in a subsequent cryptographic application such as the one-time pad (OTP)~\cite{Vernam1926,Shannon1949}.

In this paper, we consider running QKD without immediately running EC and PA.
Assuming that OTP will be used as the next step,
we delay the application of 
EC and PA
until
after the OTP, in effect performing a weakly secure OTP.
Delaying EC is trivial and requires no extra attention since errors in the original key simply translate into the same errors in the OTP-encrypted message, and bit errors do not affect the security of the message.
On the other hand, delaying PA is non-trivial since normal PA ensures a key becomes secure first before being used, and now we use the insecure key first before making it secure.
These two operations do not appear to be commuting, but we will prove that they do when we choose an appropriate PA scheme.
By commuting, we mean that
delayed PA is secure with the same security level achieved by the same PA function used to make the original raw key secure.
In summary, we prove that delaying PA after OTP does not affect the security and the key generation rate.
Delayed PA will be the focus of the paper.

At first glance, delaying privacy amplification does not appear to be of much use.
However, after a more thorough thought, we find that it is useful on at least two occasions.
First, it is useful in the secret key sharing between 
nodes in a QKD network where 
the nodes do not have a direct quantum link with each other 
but are separately connected to a common trusted relay.
QKD is run between each node and the intermediate trusted relay, without running the full QKD post-processing.
Some post-processing 
such as EC and PA
may be delayed\footnote{Investigation of the possibility of delaying basis reconciliation is beyond the scope of this paper.} until two nodes decide to share a key together, in which case, these post-processing steps are run only between them.
This is particularly useful when the classical communication, computation, and/or energy costs associated with the trusted relay are high, for example, as in satellite-based QKD~\cite{Ma_Trusted2011}.
Thus, delaying some costly post-processing parts 
can be beneficial.
In this paper, we do not discuss the trusted relay scenario 
but only the validity of delaying PA in a general manner.
Delaying EC is more trivial since any two parties each holding a bit string can remove errors between their strings by exchanging error syndromes with each other.

The second situation where delayed PA is useful is that we can use it to construct a two-way deterministic QKD protocol (DQKD) \cite{PhysRevLett.89.187902,PhysRevLett.90.157901,PhysRevLett.91.109801,PhysRevA.69.054301,Cai2004,PhysRevA.69.052319,PhysRevLett.94.140501,Lucamarini2007,Bostrom20083953}, whose security against general attacks is fully proved in this paper.
A two-way deterministic QKD protocol is a prepare-and-measure protocol in which 
each signal (in our case, a qubit) makes a round trip from Bob to Alice and back to Bob.
In contrast to conventional qubit-based QKD such as the BB84 protocol~\cite{Bennett1984}, the correct measurement basis is always used in DQKD because the signals are both prepared and measured by Bob.
To encode a key bit, Alice simply applies some operation (based on her key bit value) on the qubit sent by Bob and then returns it back to Bob.
Bob can decode Alice's key bit by a measurement in the same basis as what he used to prepare the initial qubit.
We remark that two-way deterministic QKD with continuous variables has been 
shown to have 
the potential for enhancing the security threshold~\cite{Pirandola2008}.
Delayed PA may serve as a tool for proving the security of two-way continuous-variable QKD.
In this paper, we focus on qubit-based two-way DQKD.

The security of qubit-based two-way DQKD had been a long-standing problem until our recent security proof of it~\cite{PhysRevA.84.042344}.
There, we directly compute the overall density matrix of Alice, Bob, and Eve for one particular two-way DQKD protocol in which Alice uses two operators for the encoding of her bit.
In this paper, we consider a different qubit-based two-way DQKD protocol in which Alice uses four encoding operators, and prove its security using the delayed PA idea.
We show that this particular protocol resembles 
the integration of the BB84 protocol and OTP.
Because of this, our analysis is significantly simplified
since the security of the DQKD protocol against general attacks will then directly derive from that of the BB84 protocol~\cite{Lo1999,Shor2000,Kraus2005,Renner2005,Renner2005c,Koashi2009comple} and OTP~\cite{Shannon1949}.
We simply rely on the security results of the latter.
Our proof idea is to convert the integrated scheme to the DQKD protocol through a series of equivalent protocols.
We remark that the idea of integrating QKD with OTP has been proposed before by Deng and Long~\cite{PhysRevA.69.052319} without a rigorous security analysis.
The scheme of Deng and Long runs in a batch-after-batch manner where a batch of qubits received by Alice on the BB84 channel is stored in quantum memory first before they are used as a batch for OTP encryption, in contrast to our scheme in which each qubit is returned to Bob immediately after reception by Alice.

The organization of the paper is as follows:
After reviewing some preliminaries in Sec.~\ref{sec-preliminaries},
we first prove the security of delayed privacy amplification in Sec.~\ref{sec-DPA}.
This will be an important tool that we will use in the conversion to the DQKD protocol.
The DQKD protocol is described in Sec.~\ref{sec-two-way-DQKD} and the detailed discussion of its security proof based on the conversion argument is explained in Sec.~\ref{sec-security-proof-DQKD}.
To begin the conversion, we outline the initial protocols of the conversion process in Sec.~\ref{sec-original-protocols}.
Then we discuss the conversion process in Sec.~\ref{sec-conversion}.
We conclude in Sec.~\ref{sec-conclusions}.

\section{Preliminaries}
\label{sec-preliminaries}

\subsection{Notations}
Bit strings are represented as vectors with elements in GF($2$), where GF($q$) is the Galois field with $q$ elements.
We use $\vec{k}$ to denote such a vector 
and $k[i]$ to denote the $i$th bit.
We define the projector function $P(\ket{\phi})\equiv\ket{\phi}\bra{\phi}$.
We denote the Pauli matrices by
\begin{eqnarray*}
I=
\begin{bmatrix}
1&0\\
0&1
\end{bmatrix},
\hspace{12pt}
X=
\begin{bmatrix}
0&1\\
1&0
\end{bmatrix},
\hspace{12pt}
Z=
\begin{bmatrix}
1&0\\
0&-1
\end{bmatrix},
\hspace{12pt}
Y=iXZ.
\end{eqnarray*}
and the corresponding eigenstates by $\{\ket{0_w},\ket{1_w}\}$ where $w=x,y,z$.

\subsection{Security measure}

We adopt the universal composability definition of security first proposed by Canetti~\cite{Canetti2001}.
This definition quantifies the security of a cryptographic primitive in terms of its deviation from the ideal functionality.
The notion of universal composability has been extended to the QKD setting~\cite{Ben-Or2004,BenOr:Security:05}.
\begin{definition}[\cite{Konig2007,Renner2005c,Renner_Thesis_05}]
\label{def-composability}
{\rm
A classical random variable $K$ (representing the key) drawn from the set $\mathcal{K}$ is said to be $\epsilon$-secure with respect to an eavesdropper holding a quantum system $E$ if
\begin{eqnarray}
\label{eqn-composability-def}
\frac{1}{2} \operatorname{Tr} | \rho_{KE} - \rho_{U} \otimes \rho_{E} | \leq \epsilon
\end{eqnarray}
where 
$\rho_{KE}=\sum_{k \in \mathcal{K}} P_K(k) \ket{k}\bra{k} \otimes \rho_{E|K=k}$ is the state of the systems $K$ and $E$, $P_K(k)$ is the probability of having $K=k$, 
$\rho_{U}=\sum_{k \in \mathcal{K}}  \ket{k}\bra{k} / |\mathcal{K}|$
represents an ideal key taking values uniformly over $\mathcal{K}$,
and $|\mathcal{K}|$ is the size of $\mathcal{K}$.
Here, $\operatorname{Tr} |A|=\sum_i |\lambda_i|$ where $\lambda_i$'s are the eigenvalues of $A$.
}
\end{definition}
QKD expands a shorter secret key to a longer one.
When one round of QKD that is $\epsilon_1$-secure expands on a key generated by a previous round of QKD that is $\epsilon_2$-secure, the composition of the two rounds is
$(\epsilon_1+\epsilon_2)$-secure~\cite{BenOr:Security:05}.

\subsection{Additive functions}

In this paper, we consider PA functions that are additive.
A function $f: \text{GF}(q)^N \rightarrow \text{GF}(q)^{N_\text{PA}}$ is said to be additive if 
$f(\vec{a}+ \vec{b})=f(\vec{a})+ f(\vec{b})$ for all $\vec{a},\vec{b}\in \text{GF}(q)^N$ and the addition operators are defined in the respective fields. 
For prime $q$, the function $f$ is additive if and only if it can be expressed in the form of matrix multiplication $f(\vec{a})=A\vec{a}$ where addition and multiplication are defined in $\text{GF}(q)$.
The ``if'' part of this statement is obvious.
To prove the ``only if'' part, note that
additivity implies linearity when $q$ is prime since every $\vec{a}\in \text{GF}(q)^N$ can be expressed as a pure summation of unweighted basis vectors of some basis (e.g., a weighting of $2$ is broken down into a sum of two terms as in $\vec{a}=2 \vec{v}=\vec{v}+\vec{v}$, with $\vec{v}$ being a basis vector) and $f(\vec{a})=\sum_{i=1}^N f(a[i] \vec{v}_i)=\sum_i a[i] f(\vec{v}_i)$, where $\vec{v}_i$'s are the basis vectors.
Thus, the columns of $A$ are $f(\vec{v}_i)$'s.

We work in $\text{GF}(2)$ throughout the paper and so $q=2$ and we can always consider PA functions of the form $f(\vec{a})=A\vec{a}$.
For $q=p^r$ being a power of a prime, 
we can also restrict additive $f$ to be of the same form in the following sense.
We may view $\text{GF}(q)$ as a vector space over its subfield $\text{GF}(p)$, since
any element of $\text{GF}(q)$ can be written in the form $\sum_{i=1}^r \gamma_i \beta_i$, where $\gamma_i \in \text{GF}(p)$, $\beta_i \in \text{GF}(q)$, and $\beta_i$ cannot be expressed in the form $\gamma \beta_j$ for $j \neq i$ and some $\gamma \in \text{GF}(p)$.
This means that an element of $\text{GF}(q)$ can be represented as a length-$r$ vector with elements $\gamma_i$.
Thus, when we express the PA function in the prime field so that $f: \text{GF}(p)^{rN} \rightarrow \text{GF}(p)^{rN_\text{PA}}$, the statement that $f$ is additive if and only if $f(\vec{a})=A\vec{a}$ for all $\vec{a}\in \text{GF}(p)^{rN}$ also holds.
Of course, this does not mean that $f$ can be expressed as $A\vec{a}$ for all $\vec{a}\in \text{GF}(p^r)^{N}$.

Note that the number of additive PA functions grows exponentially as $N$ increases.
This makes Eve's job to attack a key distribution scheme more difficult with larger $N$ as she has to make guesses on the PA function to be used by Alice and Bob in order to customize her attack.
Also, additivity is not a very strong constraint and 
additive functions are commonly used.
For example, in
Toeplitz matrix based PA~\cite{Mansour:Toeplitz:93,Krawczyk:Hash:94,Fung2010_Finite,Ma2011_Finite}, the PA function $f(\vec{a})=A\vec{a}$ is additive where $A$ is a Toeplitz matrix.

A property of an additive function $f$ is that any image of $f$ is a translation of the kernel by some offset.
This means that the number of elements in every image is the same.
We will this property in the proof of Theorem~\ref{thm-security-delayed-PA}.

\section{Delayed privacy amplification}
\label{sec-DPA}

\begin{figure}
\subfigure[BB84 with normal PA and OTP.]{
\includegraphics[width=1\columnwidth]{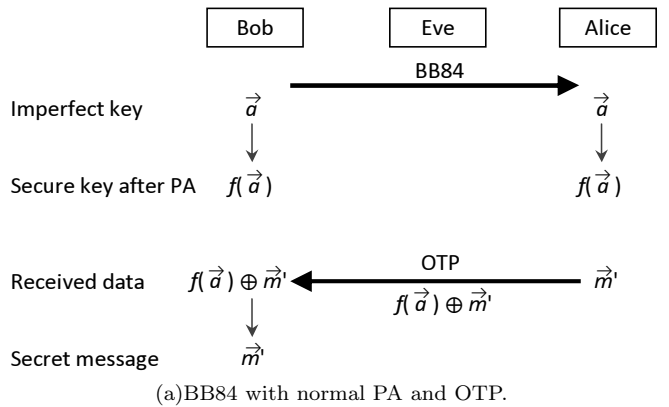}\label{fig:DPA_org}}
\\
\subfigure[BB84 with PA delayed after the OTP.  $\vec{m}$ can be regarded as the expanded version of the secret message $\vec{m}'$.]{
\includegraphics[width=1\columnwidth]{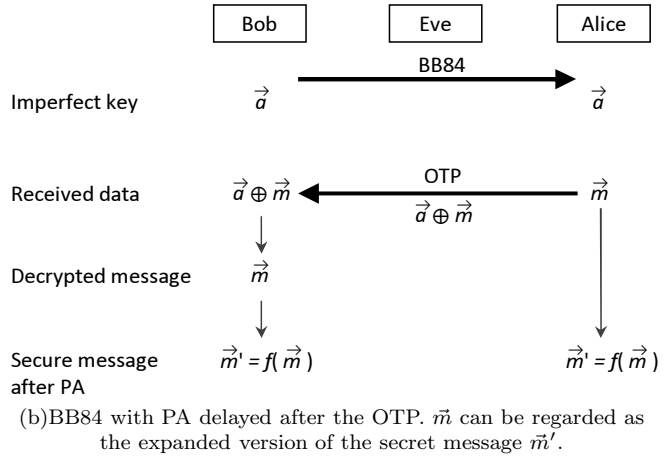}\label{fig:DPA_new}}
\caption{\label{fig:DPA}
Overview of delayed privacy amplification.
The message $\vec{m}'$ is secure in both situations, even though
Eve sees different strings in the OTP channel.
We show in Theorem~\ref{thm-security-delayed-PA} that security is not affected by whether Eve sees $f(\vec{a}) \oplus \vec{m}'$ or $\vec{a} \oplus \vec{m}$.
Here, $f$ is the PA function which shortens its input and the arrow of the BB84 channel indicates the direction of the qubits.
}
\end{figure}

Suppose Alice has an $N$-bit raw key $\vec{a}$ on which Eve has some information.
The raw key, which may not be completely secure, can be turned into a shorter secure final key by applying PA.
Bob initially holds an $N$-bit raw key $\vec{b}$ which is a noisy version of $\vec{a}$ and for the current discussion of delayed PA, we assume that Bob can correct all errors so that he also holds $\vec{a}$.
We denote the function chosen by Alice and Bob for PA as $f$ (mapping $N$ bits to $N_\text{PA}<N$ bits) and the secure key shared between Alice and Bob as $f(\vec{a})$.
Normally, to encrypt an $N_\text{PA}$-bit message $\vec{m}'$, Alice computes $f(\vec{a}) \oplus \vec{m}'$ and 
sends it to Bob.
Bob can recover the original message $\vec{m}'$ by XORing\footnote{Exclusive OR (XOR), denoted by $\oplus$, is an operation on two bits, such that $i \oplus j=0$ if both $i=j$ and $i \oplus j=1$ otherwise.  XOR can be extended to become an operation on two bit strings by XORing each bit pair independently.} the encrypted message with the shared key $f(\vec{a})$ (see Fig.~\ref{fig:DPA_org}).
Eve, in the middle of Alice and Bob, can see the encrypted message $f(\vec{a}) \oplus \vec{m}'$, but cannot get information on the original message $\vec{m}'$ because she does not know $f(\vec{a})$.

In a delayed PA scheme, 
Alice expands
the original message $\vec{m}'$ 
to $\vec{m}$ for encryption with the pre-PA key $\vec{a}$.
We call the expanded message $\vec{m}$ the PA-inverse of $\vec{m}'$.
To do this securely, as we show below, Alice should choose $\vec{m}$ (an $N$-bit string) uniformly among all strings that satisfy $\vec{m}'=f(\vec{m})$.
Alice then sends $\vec{a} \oplus \vec{m}$ to Bob (see Fig.~\ref{fig:DPA_new}).
We demand that $f$ be additive.
Thus, anyone who receives this string can apply $f$ to get $f(\vec{a} \oplus \vec{m})=f(\vec{a}) \oplus \vec{m}'$, which is the encrypted message sent in the normal OTP situation.
In particular, Bob can recover the original message by applying $f$ and XORing with the shared key $f(\vec{a})$.
Alternatively, Bob can recover the original message by XORing the received data with the pre-PA key $\vec{a}$ and applying $f$.
The security of the original message $\vec{m}'$ is not obvious, as Eve sees $\vec{a} \oplus \vec{m}$ in this delayed PA scheme, not $f(\vec{a} \oplus \vec{m})$ in the normal OTP scheme.
Nevertheless, we show in the following theorem that the security of the delayed PA scheme is identical to that of the normal OTP scheme.

\begin{theorem}[Security of delayed privacy amplification]
\label{thm-security-delayed-PA}
{\rm 
Given an additive function $f$ that maps an $N$-bit string to an $N_\text{PA}$-bit string and that
the final key $f(\vec{a})$ for some $\vec{a}$ is $\epsilon$-secure against Eve according to security Definition~\ref{def-composability},
then, for some $N_\text{PA}$-bit message $\vec{m}'$ chosen independently of $\vec{a}$,
$\vec{m}'$ is $\epsilon$-secure against Eve (i.e., with the same security level) when she sees $\vec{a} \oplus \vec{m}$, where
$\vec{m}$ is uniformly chosen among all strings that satisfy $\vec{m}'=f(\vec{m})$.
}
\end{theorem}
\begin{proof}
Due to the security of OTP,
$\vec{m}'$ is secure when Eve sees $f(\vec{a})\oplus \vec{m}'$.
Starting with this condition, we convert it to the final condition claimed in the theorem.
Let $f^{-1}[\vec{\alpha}]$ be the inverse image of $\vec{\alpha}$ under $f$.
Thus, Eve seeing 
$f(\vec{a})\oplus \vec{m}'$
is effectively the same as Eve seeing 
$f^{-1}[f(\vec{a})\oplus \vec{m}']$ 
since Eve knows $f$ and thus can compute the former given the latter and vice versa.
First note that 
$f^{-1}[\vec{\alpha}]$
has 
$2^{N-N_\text{PA}}$ elements
irrespective of $\vec{\alpha}$.
This is because
$f$ is additive and has the form $f(\vec{a})=A\vec{a}$ where $A$ is an $N_\text{PA} \times N$ matrix.
So, every inverse image is in fact an affine subspace that can be translated to the
kernel of $f$ by an offset.
Hence, all inverse images $f^{-1}[\vec{\alpha}]$ for all $\vec{\alpha}$ have the same number of elements\footnote
{Sec.~\ref{sec-computation-PA-inverse} describes the computation of $f^{-1}[\vec{\alpha}]$ and shows explicitly how to find the $2^{N-N_\text{PA}}$ elements for a given $\vec{\alpha}$.}.
This is important since this allows us to use
a random variable $v$ independent of $\vec{\alpha}$ 
to select an element in the set
$f^{-1}[\vec{\alpha}]$.
The variable $v$ has a fixed range and is drawn uniformly.

As the final part of the argument, note that 
giving Eve a random element of 
$f^{-1}[f(\vec{a})\oplus \vec{m}']$ 
chosen according to $v$
is equivalent to Eve seeing all the elements of 
$f^{-1}[f(\vec{a})\oplus \vec{m}']$ 
and $v$.
Since $v$ is independent of the elements of 
$f^{-1}[f(\vec{a})\oplus \vec{m}']$, 
knowing $v$ gives Eve
no extra information about 
$f(\vec{a})$ or $\vec{m}'$
over what knowing 
$f^{-1}[f(\vec{a})\oplus \vec{m}']$ 
gives.
Thus, from Eve's point of view, seeing a random element of the set $f^{-1}[f(\vec{a})\oplus \vec{m}']$ is equivalent to seeing the whole set.
Choosing $\vec{s}$ uniformly in 
$f^{-1}[f(\vec{a})\oplus \vec{m}']$ 
means choosing $\vec{s}$ uniformly such that $f(\vec{s})=f(\vec{a}) \oplus \vec{m}'$ or $f(\vec{s}\oplus\vec{a})=\vec{m}'$.
By defining $\vec{m}=\vec{s}\oplus\vec{a}$, we arrive at the claim of the theorem.
\end{proof}

\begin{remark}
{
\rm
We note that
delaying privacy amplification 
does not affect the security and the key generation length.
The reason is as follows.
Theorem~\ref{thm-security-delayed-PA} proves that the same security level is achieved by delaying PA with the same PA function.
Since the PA function defines the final key length, the key generation length is not affected.
}
\end{remark}

\subsection{Special messages}
We note the following special cases:
\begin{itemize}
\item
(Random message) If the original message $\vec{m}'$ is also uniformly chosen (acting as a key), $\vec{m}$ can be uniformly chosen without regard to the condition $\vec{m}'=f(\vec{m})$.

\item
(Imperfect key as message) If the original message $\vec{m}'$ is an imperfect key, we can delay the PA of it together with $\vec{a}$.
For instance, suppose that $\vec{m}'=g(\vec{a'})$ is secure after applying the PA function $g$ to the insecure key $\vec{a'}$.
Then, the encrypting party can send 
$\vec{a} \oplus \vec{m}$, where
$\vec{m}$ is uniformly chosen among all strings that satisfy $g(\vec{a'})=f(\vec{m})$.
\end{itemize}

\subsection{Computation of the PA-inverse}
\label{sec-computation-PA-inverse}
To apply the delayed PA scheme, given the original message $\vec{m}'$, Alice needs to compute its inverse by choosing
$\vec{m}$ uniformly among all strings that satisfy $\vec{m}'=f(\vec{m})$.
Here, we offer a method that Alice can use to compute the PA-inverse $\vec{m}$.
Note that this is one possible method, there may be other methods with different efficiencies to perform the same task.

Our method goes as follows.
Since $f$ is imposed to be additive, it can be represented by a matrix multiplication in $\text{GF}(2)$, the finite field of two elements:
\begin{eqnarray}
\label{eqn-PA-Am}
\vec{m}'&=&f(\vec{m})
=
A \vec{m}
\end{eqnarray}
where $A$ is an $N_\text{PA} \times N$ matrix with entries in 
$\text{GF}(2)$.
Multiplication of two elements is AND\footnote{AND is an operation on two bits such that $i$ AND $j=1$ only when $i=j=1$.}, while addition is XOR.

We assume that the rows of $A$ are linearly independent.
Thus, we can apply row operations (XOR of two rows) based on Gaussian elimination to express $A$ in upper triangular form:
\begin{eqnarray}
A&=&
R
\begin{bmatrix}
\ast & \ast & \cdots & \cdots & \cdots & \cdots &\ast 
\\
0 & \ast &  \cdots & \cdots & \cdots & \cdots &\ast 
\\
& & \ddots
\\
0& \cdots & 0 & \ast & \cdots & \cdots & \ast 
\\
0& \cdots & 0 & 0 & \ast & \cdots & \ast 
\end{bmatrix}
\end{eqnarray}
where the last row has $N_\text{PA}-1$ zeros at the beginning and $R$ is an $N_\text{PA} \times N_\text{PA}$ matrix representing the row operations with $RR=I$.
Thus, given $R \vec{m}'$, we can find $\vec{m}$ by randomly choosing the last $N-N_\text{PA}$ elements of $\vec{m}$ and successively determining the remaining elements of $\vec{m}$ by using the triangular structure.

\subsection{Example usage: a simple relay}
Suppose Bob and Charlie want to establish a secret key, but they do not have a direct quantum link with each other.
Instead, Bob has a quantum link with Alice (a relay) who has already shared a huge supply (denoted as pool $\mathcal{P}$) of perfectly secure\footnote{We assume for simplicity that the perfectly secure key is established by face-to-face key exchange.} key bits with Charlie.
Normally,
Alice and Bob would run BB84 to generate from an $N$-bit raw key $\vec{a}$ an $N_\text{PA}<N$-bit final key $f(\vec{a})$.
With this key,
Alice can OTP-encrypt 
$N_\text{PA}$ bits from pool $\mathcal{P}$ 
and send the cipher text $f(\vec{a})\oplus \vec{m}'$ to Bob, where $\vec{m}'$ denotes the bits from the pool.
Bob then recovers the confidential message $\vec{m}'$ that Charlie knows and this completes the task of sharing a key between Bob and Charlie through Alice.
If the key $f(\vec{a})$ is $\epsilon$-secure against Eve 
according to the universally composable definition in Definition~\ref{def-composability},
the originally perfectly secure key $\vec{m}'$ now becomes $\epsilon$-secure.

Suppose that 
the costs of classical communication, computation, and/or energy of Alice are high.
PA can be costly in these aspects.
In particular, performing PA requires one party to transmit the full specification of the PA function $f(\cdot)$ to the other party.
For example, Toeplitz matrix based PA needs $N+N_\text{PA}-1$ bits to specify~\cite{Mansour:Toeplitz:93,Krawczyk:Hash:94,Fung2010_Finite,Ma2011_Finite}, which can be a big number when the block size is large.
Also, performing Toeplitz matrix based PA requires large matrix multiplication, which translates to large computation and energy needs.  These can be costly for a satellite relay, for example.
In order to reduce these costs, Alice and Bob can delay PA and turn it over to Bob and Charlie.
To illustrate the idea of delayed PA, we assume for simplicity that bit and phase error testing and error correction are performed between Alice and Bob as normal.
Based on the phase error rate, Bob as in the normal situation decides a particular PA function $f$.
But instead of telling Alice about $f$, Bob tells Charlie about it.
Now, 
with delayed PA, Alice can take $N$ bits from pool $\mathcal{P}$ and directly OTP-encrypt them with the raw key $\vec{a}$.
This $N$-bit cipher text $\vec{a}\oplus \vec{m}$ is transmitted to Bob, where $\vec{m}$ denotes the $N$ bits from the pool.
Bob recovers $\vec{m}$, which Charlie knows.
Both Bob and Charlie 
apply PA $f$ to share a final secret key $f(\vec{m})$, which has length $N_\text{PA}$ and is $\epsilon$-secure according to Theorem~\ref{thm-security-delayed-PA}.
This generates the same key as in the normal situation without delayed PA.
Note that in this example, we sacrifice more key bits between Alice and Charlie to save the communication, computation, and/or energy costs of Alice.

\section{Two-way deterministic QKD protocol}
\label{sec-two-way-DQKD}

\begin{figure}[t]
\includegraphics
[width=1\columnwidth]
{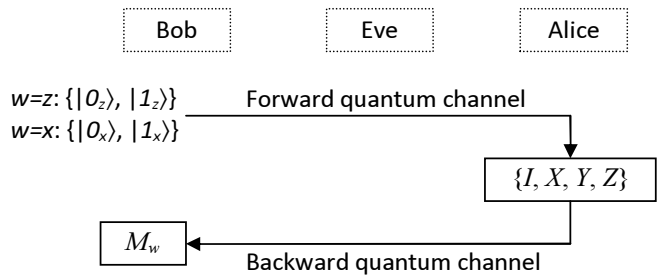}
\caption{\label{fig:DQKD_overview}
Protocol DQKD.
Bob chooses one of the four qubit states to send to Alice on the forward channel.
She applies one of the four operations to encode her key bit and returns the qubit to Bob on the backward channel.
Bob measures the received qubit in the basis he originally used for the forward qubit.
The value of Alice's key bit depends on Bob's basis (see Table~\ref{table-encoded-value}).
}
\end{figure}%

Fig.~\ref{fig:DQKD_overview} illustrates the two-way deterministic QKD protocol we consider in this paper, which we call Protocol DQKD.
The steps of Protocol DQKD are as follows.
Note that here and in the rest of paper, we 
present protocols in the context of
Koashi's security analysis~\cite{Koashi2009comple}
in which pre-shared secret keys are used for encrypted communications of error correction information.
However, paradigms of other security analyses~\cite{Lo1999,Shor2000,Kraus2005,Renner2005} are applicable as well.

\begin{enumerate}
\item (Qubit transmission) 
Bob sends 
qubits to Alice taken in $\{\ket{0_z},\ket{1_z},\ket{0_x},\ket{1_x}\}$ on line B-to-A.

\item (Encoding)
For each qubit received by Alice,
she either measures it with a random basis (check mode) or applies a random operation to it before returning it to Bob via line A-to-B (encoding mode).
We call the qubit in the check mode a test bit and in the encoding mode a code bit.
The operation she applies in the encoding mode is $I$, $X$, $Y$, or $Z$ chosen with uniform probabilities.
It does not matter whether she returns a qubit to Bob via line A-to-B in the check mode.

\item (Measurement by Bob)
Bob measures each qubit received on line A-to-B in the same basis as the one he used for the state he sent to Alice on line B-to-A in Step 1.

\item (Channel estimation)
After transmission of all qubits, Alice and Bob estimate the bit error rate $e_b$ and phase error rate $e_p$ of line B-to-A using the 
test bits measured in check mode in Step 2.
They can do this by comparing their bit values of those qubits that Alice measured with consistent bases.

\begin{table}
\begin{tabular}{|c|c|c|}
\hline
& \multicolumn{2}{c|}{Bit value} \\
\cline{2-3}
Basis & 0 & 1\\
\hline
$x$ & $\{I,X\}$ & $\{Z,Y\}$
\\
\hline
$z$ & $\{I,Z\}$ & $\{X,Y\}$
\\
\hline
\end{tabular}
\caption{Key bit value dependence on the basis used by Bob ($x$ or $z$) and Alice's encoding operation ($I$, $X$, $Y$, or $Z$).
For example, when Bob uses basis $z$, bit $1$ is encoded by Alice if she applies $X$ or $Y$ on the qubit sent by Bob.
\label{table-encoded-value}
}
\end{table}

\item (Key reconciliation)
Bob announces to Alice the basis used for each 
code bit.
Alice constructs her key bit value based on the basis (see Table~\ref{table-encoded-value}):
when the basis is $x$, the key bit is $0$ ($1$) if she applied $I$ or $X$ ($Z$ or $Y$);
when the basis is $z$, the key bit is $0$ ($1$) if she applied $I$ or $Z$ ($X$ or $Y$).
Bob uses the same rule to decide the key bit value.
Note that Alice and Bob do not discard any 
code bit.
There is no basis reconciliation step.

\item (Key bit error testing)
Alice and Bob test for the error rate $e_b^\rightleftarrows$ in the key bits by comparing a subset of them.
The remaining key bits form their raw keys, $\vec{a}$ for Alice and $\vec{b}$ for Bob.
We denote the length of them by $N$.

\item (Final key generation)
Alice and Bob choose a 
privacy amplification function $f(\cdot)$ that is additive and 
maps $N$ bits to $N_\text{PA}=N[1-h(e_p)]$ bits.
Alice applies privacy amplification to her raw key $\vec{a}$ to obtain the final key $\vec{k}=f(\vec{a})$.
She sends Bob $N h(e_b^\rightleftarrows)$ bits of error correction information encrypted with pre-shared secret bits.
This allows Bob to correct his raw key $\vec{b}$ to match Alice's $\vec{a}$.
Bob then applies PA to get the same final key $\vec{k}$.

\end{enumerate}

The net key expansion length is
\begin{eqnarray}
\label{eqn-two-DQKD-key-rate}
N_\text{key,two-way}
&=&
N[
1-h(e_b^\rightleftarrows)-h(e_p)].
\end{eqnarray}
In Sec.~\ref{sec-conversion}, we will show that the newly generated key with length $N_\text{PA}$ 
is secure, thus the net key gain given in Eq.~\eqref{eqn-two-DQKD-key-rate} is achievable.
We will prove this by combining the BB84 protocol and the OTP protocol, and then successively converting the OTP protocol to finally form the two-way DQKD protocol given here.

An interesting feature of
two-way DQKD is that every code bit encoded by Alice in the encoding mode will be used for the final key generation without being wasted due to measurement basis mismatch.
There is no basis reconciliation for the key bits and
this is why the protocol is called {\it deterministic}%
\footnote{Note that the term deterministic was first introduced in Ref.~\cite{PhysRevLett.89.187902} to mean that when Alice wants to send 0 (or 1) to Bob, she can encode her bit definitely.
This makes sense in quantum direct communication, but not QKD.
We borrow this term to the QKD setting but only use it to mean that every code bit will be used to generate the final key, instead of that every code bit is the final key bit.
This is
because Alice and Bob need to run privacy amplification which is determined only after Alice has encoded all the raw key bits.
Privacy amplification will then turn Alice's raw key bits to a new bit string that is different from what she initially sent.
}.
This is in contrast to the original BB84 protocol
where half of the code bits are discarded.
On the other hand, the efficient BB84 protocol~\cite{Lo2005b} allows all code bits to be used as well, but only asymptotically.
Therefore, in finite-length situations, two-way DQKD is still more efficient in using the code bits.

Note that the test bits in the check mode of two-way DQKD are measured by Alice with a random basis and thus are subject to discarding due to basis mismatch.
Thus, the check mode performances are the same in two-way DQKD and BB84.

A disadvantage of two-way DQKD is that the quantum signals emitted by Bob suffer from twice the channel loss compared to BB84.

\section{Security proof of two-way DQKD}
\label{sec-security-proof-DQKD}

The security proof of the two-way DQKD protocol is based on arguing for the equivalence of the protocol with an integrated scheme, and thus the security of the former directly follows from that of the latter.
The integrated scheme consists of
the BB84 protocol on the forward line and one-time pad on the backward line.
The security of both are well established%
~\cite{Lo1999,Shor2000,Kraus2005,Renner2005,Renner2005c,Koashi2009comple,Shannon1949}.
Starting with the integrated scheme in Sec.~\ref{sec-original-protocols}, we will convert it to the two-way DQKD protocol in Sec.~\ref{sec-conversion}.

\subsection{Original Protocols for constructing two-way DQKD}
\label{sec-original-protocols}
Here, we outline the steps of the BB84 protocol and OTP, which serve as the starting point of the conversion process.

\subsection*{Protocol 1 on line B-to-A: BB84}

We can view the line from Bob to Alice as a BB84 key distillation step.
The steps of BB84 are shown below, where we assume for simplicity the use of quantum memory to avoid the step of discarding bits measured with inconsistent bases.

Protocol 1 (BB84 with quantum memory) on line B-to-A:

\begin{enumerate}
\item
Bob sends $N+N_\text{test}$ qubits to Alice taken in $\{\ket{0_z},\ket{1_z},\ket{0_x},\ket{1_x}\}$.

\item
Alice stores all $N+N_\text{test}$ received qubits in quantum memory.

\item
Bob announces the basis of each qubit, and
Alice measures her qubits in the corresponding bases.

\item
Alice and Bob randomly select $N_\text{test}$ test bits to find out the bit error rate $e_b$ and phase error rate $e_p$ for this line B-to-A.\footnote{
The average quantum bit and phase error rates $e_b$ and $e_p$
are related to
the classical bit error rates in the $x$ basis test bits and the $z$ basis test bits, denoted as $e_x$ and $e_z$ respectively.
Asymptotically, the quantum bit error rate and the phase error rate for the remaining $x$ ($z$) basis bits are $e_x$ and $e_z$ respectively ($e_z$ and $e_x$ respectively).
Thus, the average quantum bit error rate is $e_b=(e_x+e_z)/2$ and the average quantum phase error rate is $e_b=(e_z+e_x)/2$.
Even though they are the same, we use separate symbols for them to emphasize their meanings in secret key distillation.
}
Alice and Bob choose a 
privacy amplification function $f(\cdot)$ that is additive and 
maps $N$ bits to $N_\text{PA}=N(1-h(e_p))$ bits.

\item

The final secret key is derived from Alice's raw key as an $N_\text{PA}$-bit string $\vec{k}=f(\vec{a})$.
To allow Bob to obtain the same final key,
Alice sends Bob $N_\text{EC}=N h(e_b)$ bits of error correction information encrypted with pre-shared secret bits of the same size so that Bob can correct his raw key $\vec{b}$ to become $\vec{a}$.
He then applies the same privacy amplification function $f(\cdot)$ to get the final key $\vec{k}$.

\end{enumerate}

Now, according to Koashi's security proof of the BB84 protocol~\cite{Koashi2009comple} (see also other proofs~\cite{Lo1999,Shor2000,Kraus2005,Renner2005}), the final key $\vec{k}$ is secure against Eve with the 
the net key expansion length as 
\begin{eqnarray}
N_\text{key}
&=&
N_\text{PA}-N_\text{EC}
=
N[
1-h(e_b)-h(e_p)].
\end{eqnarray}

\subsection*{Protocol 2 on line A-to-B: one-time pad}
We can view the line from Alice to Bob as a one-time pad encryption step.
\begin{enumerate}
\item[6.]
Alice encrypts an $N_\text{PA}$-bit message $f(\vec{m})$ with the secret key $\vec{k}$ with one-time pad and sends the encrypted message $f(\vec{m}) \oplus \vec{k}$ to Bob over a classical channel.
(As we will see later, $\vec{m}$ will be chosen randomly with uniform probabilities.)
\item[7.]
Bob decrypts his received data with key $\vec{k}$ to get the secret message $f(\vec{m})$.
\end{enumerate}
Here, 
Eve sees $f(\vec{m}) \oplus \vec{k}$ on line A-to-B and the message $f(\vec{m})$ is secure against her because of the security of one-time pad~\cite{Shannon1949}.

\subsection{Conversion from original protocol}
\label{sec-conversion}

\begin{figure}
\subfigure[]{
\includegraphics[width=1\columnwidth]{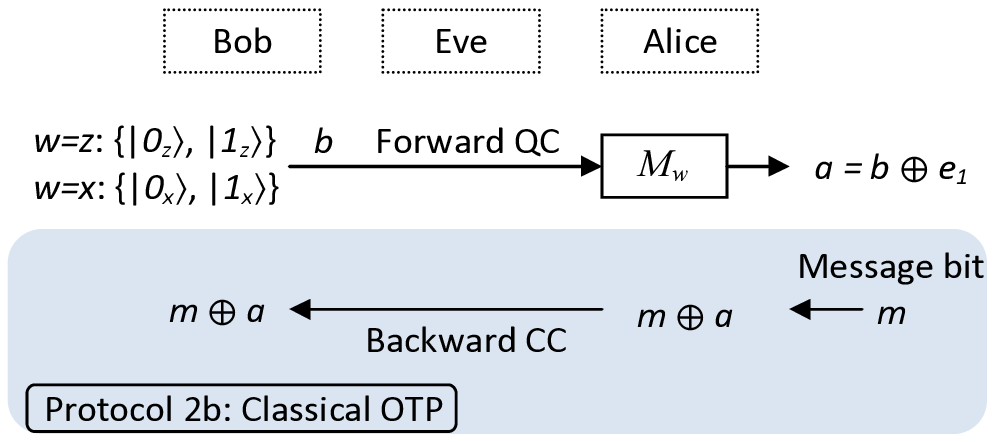}\label{fig:protocol2b}}
\\
\vspace{.2cm}
\subfigure[]{
\includegraphics[width=1\columnwidth]{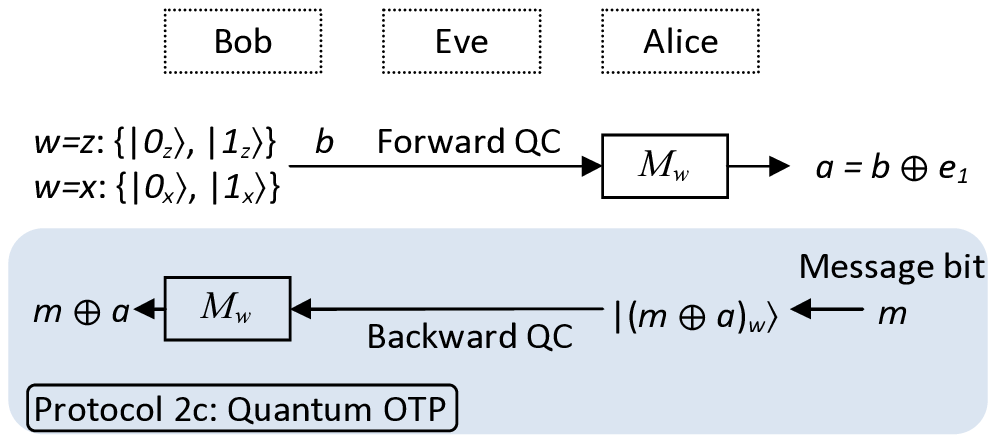}\label{fig:protocol2c}}
\\
\vspace{.2cm}
\subfigure[]{
\includegraphics[width=1\columnwidth]{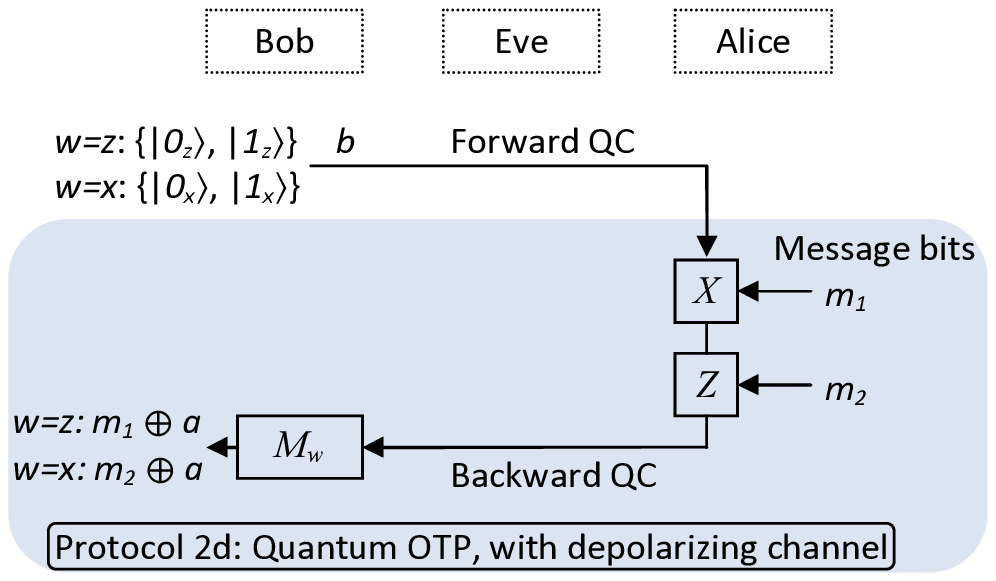}\label{fig:protocol2d}}
\caption{\label{fig:new-protocols}
(Color online)
Protocols 2b, 2c, and 2d are equivalent.
The equivalence between Protocols 2c and 2d can be understood intuitively by noting that both a measurement and a depolarizing operation disentangle line B-to-A and line A-to-B.
This is rigorously shown 
by comparing their density matrices in 
Eq.~\eqref{eqn-protocol-2c-final-state} of Protocol 2c and Eqs.~\eqref{eqn-protocol-2d-final-state-1}-\eqref{eqn-protocol-2d-final-state-2} of Protocol 2d.
We identify the message bit $m_1$ ($m_2$) of Protocol 2d with $m$ of Protocol 2c when $w=z$ ($w=x$).
Thus, Protocol 2d requires Bob to inform his basis $w$ to Alice so that she knows whether $m_1$ or $m_2$ is used for the final key generation.
In all cases, privacy amplification $f$ is delayed after OTP to generate the final secret key $f(\vec{m})$, and the top part of each figure is Protocol 1.
Here, we assume for simplicity that the backward line A-to-B is noiseless, but noise can be incorporated easily (see Sec.~\ref{subsec-key-rate}).
QC: quantum channel; CC: classical channel.
}
\end{figure}

We successively convert the original Protocol 2 to new Protocols 2b, 2c, and 2d, while maintaining the same security in each step to finally arrive at the two-way DQKD protocol.
Figure~\ref{fig:new-protocols} shows the equivalent protocols and they are described in more detail in the following.

\subsection*{
Protocol 2b on line A-to-B: one-time pad with delayed privacy amplification 
}
\begin{enumerate}
\item[6.]
Alice encrypts an $N$-bit random message $\vec{m}$ with her raw key $\vec{a}$ with one-time pad and sends the encrypted message $\vec{m} \oplus \vec{a}$ to Bob over a classical channel line A-to-B.
\item[7.]
Bob recovers the secret message $f(\vec{m})$ as follows.
He applies privacy amplification to his received bits to get $f(\vec{m} \oplus \vec{a})$.
Due to the additivity of $f(\cdot)$, his received data is $f(\vec{m}) \oplus f(\vec{a})=f(\vec{m}) \oplus \vec{k}$ which he can decrypt with the same key $\vec{k}$ to recover the message $f(\vec{m})$.
Alternatively, he can XOR his received string $\vec{m} \oplus \vec{a}$ with the raw key $\vec{a}$ and apply privacy amplification $f$ to recover the message $f(\vec{m})$.
\end{enumerate}
Here, Eve sees $\vec{m} \oplus \vec{a}$ and this is different from the situation in Protocol 2.
Nevertheless, as we have shown in Theorem~\ref{thm-security-delayed-PA}, the security of $f(\vec{m})$ is the same as that in Protocol 2, meaning that Eve cannot get any information about $f(\vec{m})$.

\subsection*{Protocol 2c on line A-to-B: one-time pad on quantum channel, with measurement}

Line A-to-B is now regarded as a quantum channel, even though
we use it for the communication of the classical OTP-encrypted message.
We encode the OTP-encrypted classical message in Step 6 of Protocol 2b $\vec{m} \oplus \vec{a}$ in a quantum state so that it can be carried by the quantum channel.
This is easily done by encoding each bit in the eigenstate of some basis.
Here, we assume that the basis used is the same basis $w=x,z$ Bob used to encode his qubit on line B-to-A.
Also, we assume for simplicity that Alice knows $w$ for each bit.
Thus, 
the $N$-qubit state Alice sends is
$\ket{(\vec{m} \oplus \vec{a})_{\vec{w}}}\equiv \bigotimes_{i} \ket{(m[i] \oplus a[i])_{w[i]}}$
where the index $i$ denotes the $i$th bit for the message, key, and basis.
The modified steps are as follows:
\begin{enumerate}
\item[6.]
Alice encrypts an $N$-bit random message $\vec{m}$ with the raw key $\vec{a}$ with one-time pad and sends the encrypted message $\ket{(\vec{m} \oplus \vec{a})_{\vec{w}}}$ to Bob over the quantum channel line A-to-B.
\item[7.]
Bob measures the $N$ qubits received from line A-to-B in basis $\vec{w}$ to recover the OTP-encrypted message $\vec{m} \oplus \vec{a}$.
He recovers the secret message $f(\vec{m})$ as in Step 7 of Protocol 2b.
\end{enumerate}

\subsubsection*{Overall density matrix}
We first consider the state $\ket{\Psi[i]}_{A\bar{A}}$ for the $i$th bit shared after Alice received her 
$N$ qubits 
from line B-to-A (where system $A$ is the $i$th qubit received by Alice on line B-to-A and system $\bar{A}$ includes all the remaining 
systems including Eve's and Bob's states for the $N$ transmissions and Alice's remaining $N-1$ qubits).
To simplify notation, we drop the index $i$ from all symbols (including $\ket{\Psi[i]}$, $w[i]$, and $a[i]$) in the following since we always deal with the $i$th qubit.
This state $\ket{\Psi}$ is the state before Alice decides to send anything on line A-to-B.
In Protocol~2c, 
Alice measures her state of $A$ in basis $w=x,z$ 
using the projection $\{\ket{0_w}\bra{0_w},\ket{1_w}\bra{1_w}\}$.
So the overall state becomes $\sum_{a=0,1} \ket{a_w}_A \ket{\Psi(a,w)}_{\bar{A}} \ket{a}_{A'}$ where 
$\ket{\Psi(a,w)}_{\bar{A}}\equiv \bra{a_w}_A \otimes I_{\bar{A}} \ket{\Psi}_{A\bar{A}}$
 and system $A'$ is the ancilla for storing the measurement result.
We specifically isolate the raw key bit $a$ in system $A'$ so that we can use it to perform OTP with the message bit $m$.
Next, Alice 
prepares a random message $2^{-1}\sum_{m=0,1} \ket{m}_{M}\bra{m}$ 
and runs controlled-$Z$ (if $w=x$) or controlled-$X$ (if $w=z$) on systems $M$ (as control) and $A$.
This is equivalent to the OTP encryption
resulting
in the overall density matrix
\begin{eqnarray*}
\rho_{MA\bar{A}A'}
&=&
\frac{1}{2}\sum_{m=0,1} \ket{m}_{M}\bra{m} \otimes 
\\
&&
P\left( \sum_{a=0,1} \ket{({m} \oplus {a})_w}_A \ket{\Psi(a,w)}_{\bar{A}} \ket{a}_{A'} \right),
\end{eqnarray*}
in which system $A$ is sent by Alice on line A-to-B to Bob and system $M$ is her message bit.

After the OTP encryption, the raw key bit $a$ is no longer needed.
Thus, we trace over system $A'$ to get the overall state
\begin{eqnarray}
\label{eqn-protocol-2c-final-state}
\rho_{MA\bar{A}}
&=&
\frac{1}{2}
\sum_{
\substack{
a=0,1\\
m=0,1}} 
P\left( \ket{m}_{M} \ket{({m} \oplus {a})_w}_A  \ket{\Psi(a,w)}_{\bar{A}}\right) .
\end{eqnarray}
Note that tracing over system $A'$ which contains the raw key bit does not mean giving the raw key bit to Eve.
Eve's state is contained in system $\bar{A}$.
The state in Eq.~\eqref{eqn-protocol-2c-final-state} is important for our discussion since it contains all the relevant systems in the protocol.
In fact, Protocol~2d in the next section will be shown to be equivalent to Protocol~2c here by showing that the corresponding states there are the same as Eq.~\eqref{eqn-protocol-2c-final-state}.

\subsection*{
Protocol 2d on line A-to-B: one-time pad on quantum channel, without measurement
}
\label{subsec-protocol-2d}

In the previous Protocol 2c, the measurement by Alice disentangles line B-to-A and line A-to-B.
Therefore, to come up with an equivalent protocol without a measurement, we need to reproduce this disentanglement feature and at the same time
achieve the same overall state in Eq.~\eqref{eqn-protocol-2c-final-state}.
One way to do this is by replacing the measurement by a depolarizing channel.
Starting with the same initial state $\ket{\Psi}_{A\bar{A}}$ for the $i$th bit as in the previous subsection, Alice performs randomly with uniform probabilities the operations, $I$, $X$, $Y$, or $Z$ on each of her $N$ qubits independently.
We express this random operation as Alice using a mixed state $4^{-1}\sum_{m_1,m_2=0,1} \ket{m_1 m_2}_{M_1 M_2} \bra{m_1 m_2}$ to control the four operations on system $A$:
\begin{eqnarray}
\ket{\Psi}_{A\bar{A}} \bra{\Psi}
& \rightarrow &
\frac{1}{4} \sum_{m_1,m_2=0,1} \ket{m_1 m_2}_{M_1 M_2}\bra{m_1 m_2} \otimes 
\nonumber
\\
&&
(X^{m_1} Z^{m_2})_A \ket{\Psi}_{A\bar{A}}\bra{\Psi} (Z^{m_2} X^{m_1} )_A
\label{eqn-protocol-2d-depolarize}
\end{eqnarray}
Here, we assume Alice holds the purification of this mixed state.
Tracing the right hand side over $M_1$ or $M_2$, simple calculations (see Appendix~\ref{app-trace-over-m1m2}) lead to
\begin{eqnarray}
\label{eqn-protocol-2d-final-state-1}
\rho_{M_1A\bar{A}}
&\negthinspace=&
\frac{1}{2}
\negthickspace\negthickspace
\sum_{
\substack{
a=0,1\\
m_1=0,1}} 
\negthickspace\negthickspace
P\left( \ket{m_1}_{M_1} \ket{({m_1} \oplus {a})_z}_A \ket{\Psi(a,z)}_{\bar{A}}\right)
\\
\label{eqn-protocol-2d-final-state-2}
\rho_{M_2A\bar{A}}
&\negthinspace=&
\frac{1}{2}
\negthickspace\negthickspace
\sum_{
\substack{
a=0,1\\
m_2=0,1}} 
\negthickspace\negthickspace
P\left( \ket{m_2}_{M_2} \ket{({m_2} \oplus {a})_x}_A \ket{\Psi(a,x)}_{\bar{A}}\right)
\end{eqnarray}
Note that Eqs.~\eqref{eqn-protocol-2d-final-state-1} and \eqref{eqn-protocol-2d-final-state-2} are expressed in terms of bases $z$ and $x$ respectively, regardless of the actual basis $w$ used by Bob to encode system $A$.

We argue that Protocol 2c and Protocol 2d are the same as follows.
Eqs.~\eqref{eqn-protocol-2d-final-state-1} and \eqref{eqn-protocol-2d-final-state-2} represent the final overall state of Protocol 2d and we compare them to that of Protocol 2c in Eq.~\eqref{eqn-protocol-2c-final-state}.
We can see that when $w=z$ ($w=x$), we can identify $m_1$ in Eq.~\eqref{eqn-protocol-2d-final-state-1} (Eq.~\eqref{eqn-protocol-2d-final-state-2}) with $m$ in Eq.~\eqref{eqn-protocol-2c-final-state}.
Thus, when $w=z$ ($w=x$), we can regard that Alice's message bit is in $m_1$ ($m_2$).
Once the basis $w$ is publicly announced by Bob, all of Alice, Bob, and Eve will know whether $m_1$ or $m_2$ will be used by Alice; 
in other words, they will know which of Eqs.~\eqref{eqn-protocol-2d-final-state-1} and \eqref{eqn-protocol-2d-final-state-2} describes the situation.
Therefore, Protocol 2c and Protocol 2d are the same from Eve's and Bob's points of view.

In Protocol 2d, we need a step for Bob to inform Alice about $w$ so that she knows whether $m_1$ or $m_2$ is her message bit.
Note that when Alice performs one of the four operations of the depolarizing channel, she does not know what the message bit value is (unless $m_1=m_2$).
After Bob receives his qubit, he announces to Alice his basis choice $w$
and only then does Alice know the value of her own message bit.

The modified steps of Protocol 2d are as follows:
\begin{enumerate}
\item[6.]
Alice chooses two $N$-bit random messages $\vec{m_1}$ and $\vec{m_2}$.
For each qubit received from line B-to-A, she applies $Z$ if $m_2=1$ and applies $X$ if $m_1=1$.\footnote{
Note that the order of applications of these two operations does not matter in light of Eq.~\eqref{eqn-protocol-2d-depolarize} as swapping the order contributes a factor of $-1$ twice.}
The qubit is then forwarded back to Bob via line A-to-B.
\item[7.]
Bob measures the $N$ qubits received from line A-to-B in basis $w$ which he has used in Step 1.
\item[8.]
Bob announces to Alice the basis for each qubit.
If the basis is $z$ ($x$), Alice's message bit is $m_1$ ($m_2$).
So in the previous step, Bob's measured qubit corresponds to Alice's OTP-encrypted message bit $m_1 \oplus a$ (for $w=z$) or $m_2 \oplus a$ (for $w=x$).
He recovers the secret message $f(\vec{m})$ as in Step 7 of Protocol 2b, with the appropriate substitution $m_1 \rightarrow m$ or $m_2 \rightarrow m$ for each bit.
Alice also constructs the secret message $f(\vec{m})$ with the same substitution.
\end{enumerate}

Therefore, when line A-to-B is noiseless, Alice and Bob will share the message bit $m_1$ (for $w=z$) or $m_2$ (for $w=x$).
When line A-to-B is noisy, we can add further error testing and error correction for $\vec{m}$, which we have omitted for simplicity of discussion.
Finally, we note that combining Protocol 1 and Protocol 2d essentially gives Protocol DQKD given in Sec.~\ref{sec-two-way-DQKD}.\footnote{The key bit error testing of Step 6 of Protocol DQKD is omitted in Protocol 2d for simplicity of discussion, but this step can easily be added without affecting the result.}
Thus, we have proved the security of Protocol DQKD.

We remark that it makes sense that Alice's message bit depends on the basis used by Bob.
Because when Bob initially sends, for example, a $z$-eigenstate to Alice via line B-to-A, only Alice's $X$ operation (controlled by $m_1$) will bit flip the state, and so $m_1$ should become her message bit.

\subsection{Key generation rates}
\label{subsec-key-rate}

In Protocols 2, 2b, 2c, and 2d, 
we assume that the final key is derived from applying privacy amplification to Alice's raw key: $\vec{k}=f(\vec{a})$.
Bob is responsible for correcting his raw key to match Alice's.
To ensure security, Alice's message $\vec{m}$ is shortened to $N_\text{PA}=N(1-h(e_p))$ bits of secure message $f(\vec{m})$, where $e_p$ is obtained from Step 4 of Protocol 1.
In the discussion so far, we have not considered errors on line A-to-B.
Errors on line B-to-A cause Alice's raw key to be different from Bob's raw key such that $\vec{b}=\vec{a}\oplus\vec{e_{1}}$, where $\vec{e_{1}}$ is the error pattern with an error rate of $e_b$ (cf. Step 4 of Protocol 1).
Errors on line A-to-B cause Bob to receive $\vec{m}\oplus\vec{a}\oplus\vec{e_2}$ in Step 6 of Protocol 2b, where $\vec{e_{2}}$ is the error pattern on this line and could be correlated with $\vec{e_{1}}$.
Thus, when Bob uses $\vec{b}$ to decrypt his message received on line A-to-B, he faces the error pattern $\vec{e_1}\oplus\vec{e_2}$, whose error rate we denote as $e_b^\rightleftarrows$.
To help Bob correct for this error pattern,
Alice sends to Bob error correction information encrypted with $N(e_b^\rightleftarrows)$ bits of the pre-shared secret key.
Therefore, the net key expansion length is
\begin{eqnarray}
N_\text{key,two-way}
&=&
N[
1-h(e_b^\rightleftarrows)-h(e_p)].
\end{eqnarray}
This represents the key generation rate for the integration of Protocol 1 and any of Protocols 2, 2b, 2c, and 2d.
As expected, this is the same formula for the two-way DQKD protocol given in Eq.~\eqref{eqn-two-DQKD-key-rate} .
As a special case, when the error rates on the two lines are both $e_b$,
the overall error rate $e_b^\rightleftarrows$ is upper-bounded by $2e_b$ since the errors on the two lines can be correlated.
Thus, the key generation rate in this case is $1-h(2e_b)-h(e_p)$, which is the same as that derived for another two-way DQKD protocol in Ref.~\cite{PhysRevA.84.042344} (see Sec.~III~F therein).

Note that as Alice and Bob are correcting the overall error pattern $\vec{e_1}\oplus\vec{e_2}$, they do not need to separately correct for the error $\vec{e_1}$ in their raw keys; i.e., they do not perform Step 5 of Protocol 1.
This is reflected in the original two-way DQKD protocol in Sec.~\ref{sec-two-way-DQKD}.

\section{Conclusions}
\label{sec-conclusions}

The central idea of our paper is delayed privacy amplification and we have proved  its security.
Delayed PA is useful for 
secret key sharing between nodes of a QKD network assisted by trusted relays, 
and for the security proof of a qubit-based two-way DQKD protocol.
We anticipate that delayed PA will have further uses in other applications, such as the security proof of two-way continuous-variable QKD~\cite{Pirandola2008}.

In this paper, we derived the qubit-based two-way DQKD protocol from an integration of the BB84 protocol and OTP,
with the condition that
PA is delayed after the one-time pad.
Because of our security proof of delayed PA, the original security of BB84 directly carries over to the DQKD protocol.
Thus, we have proved the security of the DQKD protocol against general attacks with qubit signals.
This illustrates the power of the delayed PA idea.

Security analysis of DQKD with 
multi-qubit signals and 
decoy states~\cite{Hwang2003,Lo2005,Wang2005a} is beyond the scope of the current paper and will be left for future work.
Also, using non-additive privacy amplification functions in delayed privacy amplification will be considered in the future.

\section*{Acknowledgments}%
We thank H.-K. Lo for enlightening discussions.
This work is supported in part by
RGC under Grant No. 700709P of the HKSAR Government,
NSERC, the CRC program, CIFAR, QuantumWorks,
and
NSFC under Grant No. 11074283.

\appendix

\section{
Derivation of
Eq.~\eqref{eqn-protocol-2d-final-state-1}
from Eq.~\eqref{eqn-protocol-2d-depolarize}}
\label{app-trace-over-m1m2}

We only derive 
Eq.~\eqref{eqn-protocol-2d-final-state-1}
from Eq.~\eqref{eqn-protocol-2d-depolarize}.
The derivation of 
Eq.~\eqref{eqn-protocol-2d-final-state-2} is similar.
First, we decompose $\ket{\Psi}_{A\bar{A}}=\sum_{i=0,1} \lambda_i \ket{i_z}_A \ket{e_i}_{\bar{A}}$ where $\ket{i_z}_A$ are the normalized eigenstates of basis $z$ and 
$\ket{e_i}_{\bar{A}}$ are normalized but not necessarily orthogonal.
We trace the state of Eq.~\eqref{eqn-protocol-2d-depolarize} over $M_2$:
\begin{widetext}
\begin{eqnarray*}
\rho_{M_1A\bar{A}}
&=&
\frac{1}{4} \sum_{m_1=0,1} \ket{m_1}_{M_1}\bra{m_1} \otimes 
\left[
P(
X^{m_1}_A \ket{\Psi}_{A\bar{A}}
)
+
P(
(X^{m_1} Z)_A \ket{\Psi}_{A\bar{A}}
)
\right]
\\
&=&
\frac{1}{4} \ket{0}_{M_1}\bra{0} \otimes 
\Big[
P\big(\lambda_0 \ket{0_z}_A \ket{e_0}_{\bar{A}}+
\lambda_1 \ket{1_z}_A \ket{e_1}_{\bar{A}}
\big)
+
P\big(\lambda_0 \ket{0_z}_A \ket{e_0}_{\bar{A}}-
\lambda_1 \ket{1_z}_A \ket{e_1}_{\bar{A}}
\big)
\Big]
+
\\
&&
\frac{1}{4} \ket{1}_{M_1}\bra{1} \otimes 
\Big[
P\big(\lambda_0 \ket{1_z}_A \ket{e_0}_{\bar{A}}+
\lambda_1 \ket{0_z}_A \ket{e_1}_{\bar{A}}
\big)
+
P\big(\lambda_0 \ket{1_z}_A \ket{e_0}_{\bar{A}}-
\lambda_1 \ket{0_z}_A \ket{e_1}_{\bar{A}}
\big)
\Big]
\\
&=&
\frac{|\lambda_0|^2}{2}
\Big[
\ket{0}_{M_1}\bra{0}
\otimes 
\ket{0_z}_A \bra{0_z}
+
\ket{1}_{M_1}\bra{1}
\otimes 
\ket{1_z}_A \bra{1_z}
\Big] \otimes
\ket{e_0}_{\bar{A}} \bra{e_0}
+
\\
&&
\frac{|\lambda_1|^2}{2}
\Big[
\ket{0}_{M_1}\bra{0}
\otimes 
\ket{1_z}_A \bra{1_z}
+
\ket{1}_{M_1}\bra{1}
\otimes 
\ket{0_z}_A \bra{0_z}
\Big] \otimes
\ket{e_1}_{\bar{A}} \bra{e_1}.
\end{eqnarray*}
\end{widetext}
The last equation is equal to Eq.~\eqref{eqn-protocol-2d-final-state-1} by noting that 
$\ket{\Psi(k,z)}_{\bar{A}}\equiv \bra{k_z}_A \otimes I_{\bar{A}} \ket{\Psi}_{A\bar{A}} = \lambda_k \ket{e_k}_{\bar{A}}$.

The derivation of 
Eq.~\eqref{eqn-protocol-2d-final-state-2} from Eq.~\eqref{eqn-protocol-2d-depolarize} can be done in a similar manner by 
decomposing $\ket{\Psi}_{A\bar{A}}$ in the $x$ basis as $\ket{\Psi}_{A\bar{A}}=\sum_{i=0,1} \lambda'_i \ket{i_x}_A \ket{e'_i}_{\bar{A}}$.

\bibliographystyle{apsrev4-1}

\bibliography{paperdb}

\end{document}